\documentclass[11pt]{article}
\usepackage[letterpaper, left=1.5in, right=1.5in, top=1.5in, bottom=1.5in]{geometry}
\usepackage[rm,bf]{titlesec}
\usepackage{algorithm}
\usepackage{algpseudocode}
\usepackage[all]{nowidow}
\usepackage{float}
\usepackage{enumitem}

\usepackage{amsmath}
\usepackage{amsthm}
\usepackage{multirow}
\usepackage{amssymb}

\usepackage{graphicx}
\usepackage{amssymb}
\usepackage{amsfonts}
\usepackage{amsmath,amsthm}
\usepackage[footnotesize,bf]{caption}
\usepackage{natbib}
\usepackage{setspace}
\usepackage[dvipsnames,usenames]{color}
\usepackage{pstricks,pstricks-add}
\usepackage{url}
\usepackage[usenames]{color}
\usepackage[breaklinks, colorlinks, citecolor=blue, linkcolor=blue, urlcolor=blue]{hyperref}
\usepackage{lscape}

\newtheorem{theorem}{Theorem}

\newtheorem{example}{Example}

\newtheorem{lemma}{Lemma}

\newcommand{\Mcal}{\mathcal{M}}
\newcommand{\Pcal}{\mathcal{P}}

\newcommand{\Rcal}{\mathcal{R}}

\newcommand{\Scal}{\mathcal{S}}

\newcommand{\Ccal}{\mathcal{C}}

\newcommand{\tops}{\textrm{Top}}

\newpsobject{grilla}{psgrid}{subgriddiv=1,griddots=10,gridlabels=6pt}

\begin{document}
\bibliographystyle{elsart-harv}
\title{Balanced House Allocation\footnote{We are grateful to Vikram Manjunath for insightful discussions. We also thank Eun Jeong Heo, Parag A. Pathak,  Qianfeng Tang, Guoqiang Tian, and the audience at SUFE Shanghai Micro Workshop for their comments and suggestions.  This paper is based on the third chapter of Long's Ph.D. Dissertation \citep{Long-2016}. Any errors are our own.}}

\date{\today}

\author{Xinghua Long\thanks{Institute for Advanced Research, Shanghai University of Finance and Economics, Shanghai 200433, China; \href{mailto:long.xinghua@mail.shufe.edu.cn} {long.xinghua@mail.shufe.edu.cn}}\ \ and Rodrigo A. Velez\thanks{Department of Economics, Texas A\&M University, College Station, TX 77843;
\href{mailto:rvelezca@tamu.edu}{rvelezca@tamu.edu}; \href{https://sites.google.com/site/rodrigoavelezswebpage/home}
{https://sites.google.com/site/rodrigoavelezswebpage/home}} \\\small{\textit{}}}
\maketitle

\begin{abstract}
We introduce \emph{balancedness} a fairness axiom in house allocation problems. It requires a mechanism to assign the top choice, the second top choice, and so on, on the same number of profiles for each agent. This axiom guarantees equal treatment of all agents at the stage in which the mechanism is announced when all preference profiles are equally likely. We show that, with an interesting exception for the three-agent case, Top Trading Cycles from individual endowments is the only mechanism that is balanced, efficient, and group strategy-proof.
\medskip

\medskip

\medskip

\end{abstract}

\section{Introduction}
We consider the allocation of a collective endowment of indivisible goods, objects, when monetary compensation is not available to accommodate conflicts of preference \citep{Hylland-Zeckhauser-1979}. We formulate \emph{balancedness}, a fairness axiom for this problem. It articulates an ideal symmetry notion implemented from the veil of ignorance of a policy maker who chooses the mechanism but is unaware of actual conflicts of preference. Our main result is to prove that, essentially, the celebrated Top Trading Cycles from individual endowments \citep{Shapley-Scarf-1974} are the only mechanisms that are balanced, efficient, and group strategy-proof.

It is instructive to consider a simple three-agent problem. Suppose that three students, Ann, Bob, and Charlie need to be assigned an office in an academic department. The administrator in charge will ask students for their preferences and make an assignment based on this information. The administrator is concerned about the transparency of the process and a possible perceived bias of their choices. Indeed, they may need to announce the protocol by which the allocation will take place. Moreover, randomization may be feasible at the stage of designing the protocol, but may not be feasible after the students report their preferences. It may also be complex to explain how a previous randomization among biased mechanisms provides a fair treatment of the students, when its actual realization does not.

To avoid the perception of bias in the allocation, we propose that the administrator select a mechanism that treats students symmetrically at the stage in which the mechanism is announced but preferences are not reported yet. As a benchmark we propose to assume that all preferences are equally likely. If a mechanism is unbiased, it should induce the same distribution on the ranking of assignments to each participant. That is, the number of preference profiles for which each agent gets their top choice, their second top choice, and so on, is the same for all agents. When this is so we say that the mechanism is balanced.

The perfect storm of constraints that leads to our formulation of balancedness is not far fetched.  In the related problem of school choice\emph{}, \citet{Pathak-Sethuraman-2011} describe how policy makers were concerned about the perception of fairness of mechanisms for the allocation of school seats in NY City. One option being considered was to run a lottery determining a single ranking of all students. Then students would choose their best option in that order. Another option was to run lotteries independently for each school and then use a protocol that gives agents the right to attend the school for which they have priority. As reported by \citet{Pathak-Sethuraman-2011}, the following quote by a policy maker summarizes their view:

\begin{quotation}
\emph{I cannot see how the children at the end of the line are not
disenfranchised totally if only one run takes place. I believe
that one line will not be acceptable to parents. When I
answered questions about this at training sessions, (it did
come up!) people reacted that the only fair approach was to
do multiple runs.}
\end{quotation}

Our running example illustrates the issue. Suppose that the administrator chooses at random an order of students and it is realized to be Ann is first, Bob second, and Charlie third. Even though at design stage all students are treated equally, after the order is realized Ann has a clear advantage over the other two students. Ann will receive her best office. There are possible preferences for which Bob, Charlie, or neither of them receive their best office. Bob also has an advantage over Charlie. Bob will never receive the office he considers worst. This is not so for Charlie.

It is not difficult to find a mechanism that is balanced. For instance, if the administrator chooses a given allocation and implements it independently of preferences, the mechanism is balanced. This leads to gross inefficiency, however. Obviously, balancedness does not make a mechanism appealing by itself.

The question we pose is to describe the balanced mechanisms that also satisfy basic efficiency and incentive properties. We follow \citet{Pycia-Unver-TE-2017} and \citet{Bade-MOR-2020} and require mechanisms be (Pareto) efficient and group strategy-proof.\footnote{Group strategy-proofness requires that no group of agents can lie about their preferences and guarantee a better outcome for at least one of them without decreasing the welfare of any member of the group.} Our main result, Theorem~\ref{Th:Main}, characterizes the family of mechanisms satisfying these three properties.

It turns out that the only balanced, efficient, and group strategy-proof mechanisms are essentially the Top Trading Cycles from individual endowments. These mechanisms operate as follows. The administrator endows each agent with an object and then asks them to trade following the Top Trading Cycles protocol (see Sec.~\ref{Sec:Results} for details). If the number of agents is different from three, these are the only mechanisms satisfying these properties. Curiously, when there are exactly three agents there is an additional class of mechanisms satisfying our desiderata, the Trading Cycles mechanisms with three brokers \citep{Pycia-Unver-TE-2017,Bade-MOR-2020}. In these mechanisms, the administrator assigns a weak form of property rights over the objects to allocate, which is referred to as brokerage in the literature. The administrator commits to an efficient assignment and minimizes the number of agents who receive the object they broker. Ties are resolved by identifying the agent who brokers the most popular object. If this agent competes for this object with exactly another agent, ties are resolved in favor of the other agent. Otherwise, ties are resolved in favor of the agent who brokers the most popular object.

We assume throughout that there is the same number of agents and objects, that agents have strict preferences, and that each agent is to receive an object. This model, introduced  by \citet{Hylland-Zeckhauser-1979}, is usually referred to as house allocation. There is a significant body of literature exploring different normative considerations in this environment. The most relevant to us are those requiring efficiency and incentive properties, which make the base of our normative objectives. In particular, a mechanism is efficient, group strategy-proof, and neutral if and only if it is a serial dictatorship, i.e., for a given order of the agents, the first agent chooses her best object, the second agent chooses her best object among the remaining objects, and so on \citep{Svensson-1999-SCW}.\footnote{A mechanism is neutral if it is invariant under permutation of the name of objects.} In contrast to most other environments, there is a plethora of non-neutral mechanisms that are efficient and group strategy-proof. The first to exhibit well behaved families were \citet{ABDULKADIROGLU-Sonmez-1999} in a related model that allows initial ownership of some objects, and \cite{Papai-2000a} who introduced a general class of mechanisms based on trade from private endowments that are hierarchically revised as agents obtain their allotments. None of these families exhaust the whole class of efficient and group strategy-proof mechanisms. It was only until \citet{Pycia-Unver-TE-2017} and \citet{Bade-MOR-2020} that the complete class of mechanisms satisfying these two properties was characterized. These mechanisms generalize the hierarchical exchange mechanisms by introducing a weak form of control right in which an agent can trade an object without having the right to take it for themself.

Our results allow us to identify a dimension in which Top Trading Cycles is superior to all other efficient and group strategy-proof mechanisms. To the length of our knowledge this is the first result that articulates such a message in our basic environment of house allocation.\footnote{Top Trading Cycles from individual endowments has been found to be fundamentally different from other mechanisms in an augmented model in which objects have priorities over agents. This  is the only efficient and strategy-proof mechanism that minimizes violations of priorities \citep{Abdul-Che-2010,Morill-2013-ET,Abdul-et-al-2017-NBER}} Indeed, previous literature has concentrated on the evaluation of fairness properties of efficient and group strategy-proof mechanisms from an ex-ante perspective. Evaluated at the designing stage when the mechanism designer still has the opportunity to randomize their design parameters, Top Trading Cycles is no different from the other efficient and strategy-proof mechanisms. The basic form of this result was developed independently by \citet{Knuth-1996} and \citet{Abdulkadiroglu-Sonmez-1998}. Consider the random mechanism that is obtained by selecting at random an endowment profile and then running the Top Trading Cycles for this endowment. For a given preference profile, the distribution over matchings generated by this procedure is exactly the same as that one obtains by randomly choosing an order of the agents and then running a serial dictatorship with respect to this order. This equivalence was proven to extend to increasingly larger families of efficient and group strategy-proof mechanisms \citep{Pathak-Sethuraman-2011,Lee-Sethuraman-2011,CARROLL-2014-JME}. Remarkably, the equivalence actually holds for suitably defined symmetrized versions of \emph{all} efficient and group strategy-proof mechanisms \citep{Bade-MOR-2020}.

The closest paper to ours, at a conceptual level, is \citet{Freeman-et-al-2021}. These authors independently define order symmetry, a fairness axiom parameterized by a probability measure $P$ on the space of preferences in house allocation problems. The axiom requires the mechanism guarantee each agent the same probability to obtain their top choice, their second top choice, and so on, when the realization of preferences follows $P$. Our notion of balancedness corresponds to order symmetry with respect to the uniform distribution on the preference space.  While our objectives overlap, our results are essentially independent. \citet{Freeman-et-al-2021}'s main contribution is to show that Top Trading Cycles satisfies order symmetry with respect to symmetric probability distributions. This implies that Top Trading Cycles is balanced, a result first stated by \citet{Long-2016}, and which we include for completeness in the presentation. Our main contribution is to characterize the whole family of balanced, efficient, and group strategy-proof mechanisms. That is, to show that with an interesting exception when there are exactly three agents, Top Trading Cycles is singled out as the only balanced mechanisms among all efficient and group strategy-proof mechanisms.

The remainder of the paper is organized as follows. Sec.~\ref{Sec.model} introduces the model. Sec.~\ref{Sec:Results} presents our results. Finally, Sec.~\ref{Sec.COnclusion} concludes.

\section{Model}\label{Sec.model}

\subsection{Environment}

A set of $n$ objects $H$ is to be allocated among $n$ agents $N:=\{1,...,n\}$.  Generic objects are $x,y,...$ and generic agents $i,j$. Agents have unit demand and privately known strict preferences over the set of objects. We assume that all objects are acceptable to all agents.  We denote by $\Rcal$ the set of strict preference relations on $H$, i.e., complete, transitive, and antisymmetric binary relations. The generic preference of agent~$i$ is $R_i$ and the generic preference profile is $R=(R_i)_{i\in N}\in \Rcal^N$. The object ranked $k$-highest in $R_i$ is $R_i^k$. With this notation, the best object for $R_i$ is $R_i^n$, the second highest $R_i^{n-1}$, and so on.   The top ranked object of $R_i$ in a subset of objects $H'\subseteq H$ is $\tops(R_i,H')$. We denote the strict part of $R_i$ by $P_i$. For any group of agents $S\subseteq N$, we denote the preferences of the agents in $S$ by $R_S$.

An matching $\mu$ is a bijection between $N$ and $H$. The object assigned to agent~$i$ under $\mu$ is $\mu_i$. Let $\cal M$ be the set of all possible matching. Whenever convenient we describe a matching in its set theory form as a subset of $N\times H$, $\{(i,x),(j,y),...\}$. A submatching is a subset of a matching. The set of submatchings not in~$\Mcal$ is~$\Scal$. The generic submatching is $\nu$. The agents and objects matched by $\nu$ are~$N_\nu$ and~$H_\nu$, respectively. The complements of these sets are
$N^c_\nu:=N\setminus N_\nu$ and $H^c_\nu:=H\setminus H_\nu$.  A mechanism is a function from $\Rcal^N$ to $\cal M$. The generic mechanism is~$f$. Agent~$i$'s allotment at $R$ is $f(R)_i$.

\subsection{Axioms}

Our central axiom requires mechanisms to treat agents symmetrically from the point of view of a designer who selects the mechanism and considers agents preferences to be all equally probable.

A mechanism $f$ is \textbf{balanced} if for each pair of agents $i$ and $j$, and each $k\in \{1,...,n\}$, the number of profiles in $\Rcal^N$ for which agent $i$ receives her $k$-th ranked object is the same as that of agent~$j$, i.e., \[\left|\left\{R\in\Rcal^N: f(R)_i=R_i^k \right\}\right|=\left|\left\{R\in\Rcal^N: f(R)_j=R_j^k \right\}\right|.\]

We consider two additional basic axioms.

A mechanism $f$ is \textbf{efficient} if for each $R\in\Rcal^N$, $f(R)$ is  an \textbf{efficient} matching, i.e., there is no other matching that each agent finds at least as good as $f(R)$ and at least one agent prefers.

A mechanism $f$ is \textbf{group strategy-proof} if no group of agents can gain by misreporting their preferences in the direct revelation game associated with the mechanism. That is, for each $S\subseteq N$ and each $R\in\Pcal^N$, there is no $R_S'\in\Pcal^S$ such that for each $i\in S$, $f(R_S',R_{-S})_i\,R_i\,f(R)_i$ and there is $j\in S$ such that $f(R_S',R_{-S})_j\,P_j\,f(R)_j$.

\section{Results}\label{Sec:Results}

We characterize the set of balanced mechanisms that are efficient and group strategy-proof. When the number of objects is different from three, these mechanisms are the Top Trading Cycles from individual endowments. When there are exactly three objects, besides these mechanisms, an additional class known as TC mechanisms with three brokers also satisfy these properties.

We first define the mechanisms involved in the statement in our characterization. We then state the theorem and prove it in a series of lemmas.

\textbf{TTC from individual endowments \citep{Shapley-Scarf-1974}}: An endowment profile is a bijection $\omega:N\rightarrow H$, where agent~$i$'s endowment is~$\omega_i$. The TTC mechanism from individual endowment $\omega$, $TTC_\omega$, assigns to each agent the allotment calculated by means of the following algorithm.\footnote{The TTC algorithm was attributed to David Gale by \citet{Shapley-Scarf-1974}.} At the first step construct a graph in which each agent points to her most preferred object, and each object points to its owner. Since the number of agents and objects is finite, the graph has at least a cycle. $TTC_\omega$ assigns each agent in a cycle her most preferred object.  In the subsequent steps consider the set of agents who have not received an allotment yet. Again form the graph in which each agent points to her most preferred object still to be allocated. Agents in a cycle of this graph receive their most preferred object in the remaining objects. The algorithm continues until all agents have received an allotment.\footnote{The TTC algorithm can be run clearing all cycles that form in a step. It can also be run selecting a single cycle from those that form in a given step. The order in which the cycles are cleared does not change the outcome of the algorithm. We use this second version of the algorithm in our analysis, for it allows us to associate steps and cycles.}

\textbf{TC mechanisms with three brokers \citep{Bade-MOR-2020}}: These mechanisms are defined when there are exactly three agents and three objects. A brokerage profile is a bijection $b:N\rightarrow H$. The TC mechanism with brokerage profile $b$, $TC_b$, allocates objects efficiently while trying to minimize the number of brokers who receive the object they broker. For each $R$, let $M_b(R)$ be the set of efficient matchings that minimize the number of brokers who receive the object they broker. Note that for each $i\in N$ and each $x\in H$, there is at most an element in $M_b(R)$ in which agent $i$ receives~$x$. Thus, for a given agent, there is a unique best matching and a unique worst matching in $M_b(R)$. Now, if $M_b(R)$ is not a singleton, there is an agent who brokers an object that is at the top of the preference of at least two agents, for otherwise there would be a unique efficient matching for $R$. Then, we can divide the definition of $TC_b$ in three exclusive cases.
\begin{enumerate}
\item $M_b(R)$ is a singleton: $TC_b(R)$ is the unique element in this set.
  \item $M_b(R)$ is not a singleton, the top object of agents $i$ and $j$ is $b_i$, and the top object of agent $k$ is different from $b_i$: $TC_b(R)$ is the worst matching in $M_b(R)$ for agent~$i$.
  \item $M_b(R)$ is not a singleton and the top object of agents $j$ and $k$ is $b_i$: $TC_b(R)$ is the best matching in $M_b(R)$ for agent~$i$.
\end{enumerate}

We can now state our main result.

\begin{theorem}\rm\label{Th:Main}
A mechanism $f$ is balanced, efficient, and group strategy-proof if and only if it is a TTC mechanism from individual endowments or a TC mechanism with three brokers.
\end{theorem}

We first prove that TTC from individual endowments and TC mechanisms with three brokers are balanced. It is known that these mechanisms satisfy the other properties \citep{Shapley-Scarf-1974,Roth-1982,Bird-1984,Pycia-Unver-TE-2017,Bade-MOR-2020}.

\begin{lemma}\rm\label{Lem:balanced}TTC mechanisms from individual endowments and TC mechanisms with three brokers are balanced.
\end{lemma}

\begin{proof}Consider first a TTC mechanism from individual endowments.
Let $\omega$ be an endowment profile. Fix two different agents $i$ and $j$ in $N$.  For any preference profile $R\in \Rcal^N$, suppose $TTC_\omega(R)_i=R_i^r$, and $TTC_\omega(R)_j = R_j^l$. Construct a bijection $\tau: \Rcal^N\to \Rcal^N$ such that for each agent $h\in N\setminus\{i,j\}$, swap the positions of $\omega_i$ and $\omega_j$ in $R_h$ to form $\tau(R)_h$, while $\tau(R)_i$ and $\tau(R)_j$ are formed by swapping the positions of $\omega_i$ and $\omega_j$ in $R_j$ and $R_i$, respectively. We will show that agent~$i$ receives her $l$-th ranked object and agent~$j$ receives her $r$-th ranked object at $\tau(R)$. Consequently the mechanism is balanced.

Consider the TTC algorithm at $R$. Let $t=1,...,m$ be the steps in the algorithm. Suppose without loss of generality that there is a minimal trading cycle at each step, i.e., a cycle with no proper subcycle. Let $N_t$ be the set of agents in the cycle at step $t$ and for $t\geq 2$, $H_{t}:= H\setminus \{TTC_\omega(R)_k:k\in N_1\cup\dots\cup N_{t-1}\}$. Suppose that agent $i$ gets an object at step $t_i$ and agent $j$ gets one at step $t_j$. Suppose without loss of generality that $t_i\leq t_j$.

We calculate $TTC_\omega(\tau(R))$. Suppose first that $1<t_i$. Since $N_1$ and their endowments form a cycle at step 1 at profile $R$, for each $k\in N_1$, $\tops(R_k,H)\,P_k\,\{\omega_i,\omega_j\}$. Then, for each $k\in N_1$, $\tops(\tau(R)_k,H)=\tops(R_k,H)$ and $\tops(R_k,H)\,\tau(P_k)\,\{\omega_i,\omega_j\}$. Thus, agents $N_1$ and their endowments form a cycle at the first step of $TTC_\omega$ at profile $\tau(R)$. We select this to be the first cycle in the construction of $TTC_\omega(\tau(R))$. In the trading cycle, each agent $k\in N_1$ points to $\tops(R_k,H)$. Thus, for each $k\in N_1$, $TTC_\omega(\tau(R))_k=TTC_\omega(R)_k$. Let $T\leq t_i-2$. Suppose that we have completed the first $t\leq T$ steps in the construction of $TTC_\omega(\tau(R))$ and at each such step we proved that $N_t$ forms a cycle of the algorithm that calculates $TTC_\omega(\tau(R))$ and for each $k\in N_t$, $TTC_\omega(\tau(R))_k=TTC_\omega(R)_k$. Let $k\in N_{T+1}$. Since $N_{T+1}$ and their endowments form a cycle at step $T+1$ at profile $R$, for each $k\in N_{T+1}$, $\tops(R_k,H_{T+1})\,P_k\,\{\omega_i,\omega_j\}$. Then, for each $k\in N_{T+1}$, $\tops(\tau(R)_k,H_{T+1})=\tops(R_k,H_{T+1})$ and $\tops(R_k,H_{T+1})\,\tau(P_k)\,\{\omega_i,\omega_j\}$. Note that $H_{T+1}= H\setminus$\linebreak $\{TTC_\omega(\tau(R))_k:k\in N_1\cup\dots\cup N_t\}$.  Thus, agents $N_{T+1}$ and their endowments form a cycle at step $T+1$ of the algorithm that calculates $TTC_\omega(\tau(R))$. Each agent $k\in N_{T+1}$ points to $\tops(R_k,H_{T+1})$. Thus, for each $k\in N_{T+1}$, $TTC_\omega(\tau(R))_k=TTC_\omega(R)_k$.

\begin{center}
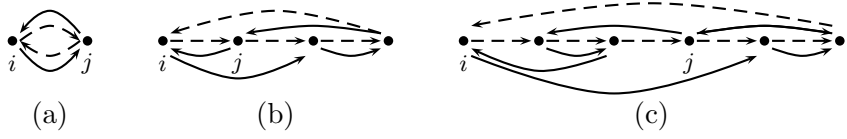
\begin{figure}[t]
\centering
\begin{pspicture}(0,-.5)(11,1.5)
\psdots(0,1)(1,1)
\rput[c](0,.7){$\mbox{\footnotesize$i$}$}
\pscurve[linestyle=dashed]{->}(0.1,1)(0.5,1.2)(.9,1)
\pscurve[linestyle=dashed]{->}(.9,1)(0.5,.8)(0.1,1)
\pscurve{<-}(0.1,1.1)(0.5,1.4)(.9,1.1)
\pscurve{<-}(.9,.9)(0.5,.6)(0.1,.9)
\rput[c](1,.7){$\mbox{\footnotesize$j$}$}
\rput[c](0.5,0){(a)}
\psdots(2,1)(3,1)(4,1)(5,1)
\psline[linestyle=dashed]{->}(2.1,1)(2.9,1)
\psline[linestyle=dashed]{->}(3.1,1)(3.9,1)
\psline[linestyle=dashed]{->}(4.1,1)(4.9,1)
\pscurve{->}(4.1,0.9)(4.5,0.8)(4.9,0.9)
\rput[c](2,.7){$\mbox{\footnotesize$i$}$}
\rput[c](3,.7){$\mbox{\footnotesize$j$}$}
\pscurve{->}(4.9,1.1)(4,1.2)(3.1,1.1)
\pscurve[linestyle=dashed]{->}(4.9,1.1)(3.5,1.4)(2.1,1.1)
\pscurve{<-}(2.1,0.9)(2.5,0.8)(2.9,0.9)
\pscurve{->}(2.1,0.8)(3,0.5)(3.9,0.8)
\rput[c](3.5,0){(b)}
\psdots(6,1)(7,1)(8,1)(9,1)(10,1)(11,1)
\psline[linestyle=dashed]{->}(6.1,1)(6.9,1)
\psline[linestyle=dashed]{->}(7.1,1)(7.9,1)
\psline[linestyle=dashed]{->}(8.1,1)(8.9,1)
\psline[linestyle=dashed]{->}(9.1,1)(9.9,1)
\psline[linestyle=dashed]{->}(10.1,1)(10.9,1)
\pscurve{->}(7.1,0.9)(7.5,0.8)(7.9,0.9)
\pscurve{->}(10.1,0.9)(10.5,0.8)(10.9,0.9)
\pscurve{<-}(6.1,0.9)(7,0.6)(7.9,0.8)
\rput[c](6,.7){$\mbox{\footnotesize$i$}$}
\rput[c](9,.7){$\mbox{\footnotesize$j$}$}
\pscurve{->}(6.1,0.8)(8,0.3)(9.9,0.8)
\pscurve{->}(9.1,1.1)(10,1.2)(10.9,1.1)
\pscurve{<-}(9.1,1.1)(10,1.2)(10.9,1.1)
\pscurve{<-}(7.1,1.1)(8,1.2)(8.9,1.1)
\pscurve[linestyle=dashed]{->}(10.9,1.2)(8.5,1.5)(6.1,1.2)
\rput[c](8.5,0){(c)}
\end{pspicture}
\caption{Cycle $N_{t+1}$ in $TTC_\omega(R)$ and $TTC_\omega(\tau(R))$. Each agent points to the agent who owns her most preferred remaining object. Dashed arrows correspond to $TTC_\omega(R)$ and solid lines to  $TTC_\omega(\tau(R))$. }\label{Figure-cycles}
\end{figure}
\end{center}
\vspace{-1cm}

Suppose first that $t=t_i=t_j$. There are three cases.

Case 1.a: $TTC_\omega(R)_i=\omega_j$ and $TTC_\omega(R)_j=\omega_i$. Then, $\tops(R_i,H_{t})=\omega_j$ and $\tops(R_j,H_{t})=\omega_i$. By definition of $\tau(R)$,  $\tops(\tau(R)_i,H_{t})=\omega_j$ and $\tops(\tau(R)_j,H_{t})=\omega_i$ (Fig.~\ref{Figure-cycles} (a)).

Case 1.b:  $TTC_\omega(R)_i=\omega_j$ and $TTC_\omega(R)_j\neq \omega_i$. Then, $\tops(\tau(R)_j,H_{t})=TTC_\omega(R)_j \not\in\{\omega_i,\omega_j\}$ and $\tops(\tau(R)_j,H_{t})=\omega_i$.  Moreover, the agents in $N_t$ and their endowments form a trading cycle in which agent $i$ receives $TTC_\omega(R)_j$ and agent $j$ receives $\omega_i$ (Fig.~\ref{Figure-cycles} (b)).

Case 1.c: $TTC_\omega(R)_i\neq \omega_j$ and $TTC_\omega(R)_j\neq \omega_i$. Then, $\tops(\tau(R)_i,H_{t})=TTC_\omega(R)_j\not\in\{\omega_i,\omega_j\}$ and $\tops(\tau(R)_j,H_{t})=TTC_\omega(R)_i\not\in\{\omega_i,\omega_j\}$.  Moreover, the agents in $N_t$ and their endowments form a trading cycle in which agent $i$ receives $TTC_\omega(R)_j$ and agent $j$ receives $TTC_\omega(R)_i$ (Fig.~\ref{Figure-cycles} (c)).

Thus, we can always select a trading cycle at step $t$ of the algorithm that calculates $TTC_\omega(\tau(R))$, in which agent~$i$ receives her $l$-th ranked object and agent~$j$ receives her $r$-th ranked object.

Suppose now that $t=t_i<t_j$. There are two cases.

Case 2.a: $TTC_\omega(R)_i=\omega_i$. Then, $N_t=\{i\}$ and $\tops(R_i,H_{t})=\omega_i$. Then,\linebreak $\tops(\tau(R)_j,H_t)=\omega_j$. We can select the cycle in which agent $j$ points to her own object as the cycle at step $t$ of the algorithm that calculates $TTC_\omega(\tau(R))$.

Case 2.b: $TTC_\omega(R)_i\neq \omega_i$. Then, at the $t$-th step of the algorithm that calculates $TTC_\omega(\tau(R))$ the agents in $\{j\}\cup (N_t\setminus\{i\})$ and their objects form a cycle in which agent $j$ points to $TTC_\omega(R)_i$.

Thus, we can always select a trading cycle at step $t$ of the algorithm that calculates $TTC_\omega(\tau(R))$, in which agent~$j$ receives her $r$-th ranked object.

Observe that for each $t_i<t<t_j$ and each $k\in N_t$, $\tops(\tau(R)_k,\{\omega_i\}\cup (H_t\setminus\{\omega_j\}))=\tops(R_k,H_t)$.  Thus, for each such $t$ the agents $N_{t}$ and their endowments form a cycle at step $t$ of the algorithm that calculates $TTC_\omega(\tau(R))$. Moreover, each agent $k\in N_{t}$ points to $\tops(R_k,H_{t})$. Thus, for each $k\in N_{t}$, $TTC_\omega(\tau(R))_k=TTC_\omega(R)_k$.

Finally, let $t=t_j$. There are two cases.

Case 3.a: $TTC_\omega(R)_j=\omega_j$. Then, $N_t=\{j\}$ and $\tops(R_j,H_{t})=\omega_j$. Then,\linebreak $\tops(\tau(R)_i,H_t)=\omega_i$. We can select the cycle in which agent $i$ points to her own object as the cycle at step $t$ of the algorithm that calculates $TTC_\omega(\tau(R))$.

Case 3.b: $TTC_\omega(R)_j\neq \omega_j$. Then, at the $t$-th step of the algorithm that calculates $TTC_\omega(\tau(R))$ the agents in $\{i\}\cup (N_t\setminus\{j\})$ and their objects form a cycle in which agent $i$ points to $TTC_\omega(R)_j$.

Thus, we can always select a trading cycle at step $t$ of the algorithm that calculates $TTC_\omega(\tau(R))$, in which agent~$i$ receives her $l$-th ranked object.

Now, consider TC mechanisms with three brokers. These mechanisms are defined for the three-agent case only. Suppose without loss of generality that $N=\{1,2,3\}$. Let $b$ and $c$ two brokerage structures. We claim first that if $TC_b$ is balanced, so is $TC_c$. Let $\pi:=b\circ c^{-1}:H\rightarrow H$. Since both $b$ and $c$ are bijections, $\pi$ is a bijection. Moreover, $\pi^{-1}=c\circ b^{-1}$. For each $R\in \Rcal^N$ let $\tau(R)$ be the defined as follows: for each $i\in N$ and $\{x,y\}\subseteq H$, $x\,\tau(R)_i\,y$ if and only if $\pi(x)\,R_i\,\pi(y)$. Thus, for each $i\in N$ and $\{x,y\}\subseteq H$, $x\,R_i\,y$ if and only if $\pi^{-1}(x)\,\tau(R)_i\,\pi^{-1}(y)$. Since $\pi$ is a bijection, so is $\tau:\Rcal^N\rightarrow\Rcal^N$. We have that:

(i) $\pi^{-1}$ and $\pi$ preserve the rankings of $R$ and $\tau(R)$, respectively. That is, for $i\in N$ and $k\in\{1,2,3\}$, $R_i^k=\pi(\tau(R)_i^k)$ and $\pi^{-1}(R_i^k)=\tau(R)_i^k$.

(ii) $\mu\in\Mcal$ is efficient for $R$ if and only if $\pi^{-1}\circ\mu$ is efficient for $\tau(R)$.

(iii) For $i\in N$ and $\mu\in\Mcal$,  $\mu_i=b_i$ if and only if $(\pi^{-1}\circ \mu)_i=c\circ b^{-1}(b_i)=c_i$. Thus, for each $\mu\in\Mcal$, $|\{i:\mu_i=b_i\}|=|\{i:(\pi^{-1}\circ \mu)_i=c_i\}|$.

We claim that $TC_c(\tau(R))=\pi^{-1}\circ TC_b(R)$. By (ii) and (iii) $M_b(R)=\pi(M_b(\tau(R)))$ and $\pi^{-1}(M_b(R))=M_b(\tau(R))$.
If $M_b(R)$ is a singleton, so is $M_b(\tau(R))$ with unique element $\pi^{-1}\circ TC_b(R)$.

By (i), if for profile $R$ the top object of agents $i$ and $j$ is $b_i$, and the top object of agent $k$ is different from $b_i$, then for profile $\tau(R)$ the top object of agents $i$ and $j$ is $c_i$, and the top object of agent $k$ is different from $c_i$. Thus, if $TC_b(R)$ is calculated as the worst matching of agent $i$ in $M_b(R)$, we have that $TC_c(\tau(R))$ is calculated as the worst matching of agent $i$ in $M_b(\tau(R))$. By (i) the worst matching of $\tau(R)_i$ in $M_b(\tau(R))$ is the composition of $\pi^{-1}$ and the worst matching of $R_i$ in $M_b(R)$. Thus, $TC_c(\tau(R))=\pi^{-1}\circ TC_b(R)$.

By (i) if for profile $R$ the top object of agents $j$ and $k$ is $b_i$, then for profile $\tau(R)$ the top object of agents $j$ and $k$ is $c_i$. Thus, if $TC_b(R)$ is calculated as the best matching of agent $i$ in $M_b(R)$, we have that $TC_c(\tau(R))$ is calculated as the best matching of agent $i$ in $M_b(\tau(R))$. By (i) the best matching of $\tau(R)_i$ in $M_b(\tau(R))$ is the composition of $\pi^{-1}$ and the best matching of $R_i$ in $M_b(R)$. Thus, $TC_c(\tau(R))=\pi^{-1}\circ TC_b(R)$.

Since $TC_c(\tau(R))=\pi^{-1}\circ TC_b(R)$ by (i), for each $i\in N$, the ranking of $TC_b(R)_i$ in $R_i$ is the same as the ranking of $TC_c(\tau(R))_i$ in $\tau(R)_i$. Since $\tau$ is a bijection, if $TC_b$ is balanced, so is $TC_c$.

One can calculate that for a given brokerage structure:
\begin{align*}
\left|\left\{R\in\Rcal^N:TC_b(R)_i=R_i^3 \right\}\right|=144, \ \forall i\in \{1,2,3\}. \\
\left|\left\{R\in\Rcal^N:TC_b(R)_i=R_i^2 \right\}\right|=48, \ \forall i\in \{1,2,3\}. \\
\left|\left\{R\in\Rcal^N:TC_b(R)_i=R_i^1 \right\}\right|=24, \ \forall i\in \{1,2,3\}.
\end{align*}
\end{proof}

To prove the other direction of Theorem~\ref{Th:Main}, we need to consider the whole space of efficient and group strategy-proof mechanisms. These mechanisms have been described by \citet{Pycia-Unver-TE-2017} and \citet{Bade-MOR-2020}. We refer the reader to these papers for the complete description of all mechanisms in this family. This extensive detail is not necessary to prove our theorem. For our purpose it is sufficient to describe the general algorithm by which they are defined and the parameters that define it.

An \textbf{inheritance structure} $c$ assigns to each $\nu\in\Scal$ and each $x\in H^c_\nu$ a pair $c_\nu(x)\in N^c_\nu\times \{o,b\}$. If $c_\nu(x)=(i,o)$ we say that agent $i$ owns $x$ at $\nu$. If $c_\nu(x)=(i,b)$ we say that agent $i$ brokers $x$ at $\nu$.  The \textbf{TC mechanism associated with inheritance structure} $c$, $TC_c$, is defined by means of the following owner-and-broker algorithm. The algorithm starts from the empty matching. If there are three brokers, it determines the assignment with a TC mechanism with three brokers. If there are less than three brokers, the algorithm proceeds as follows. It constructs a graph in which agents who are owners point to their best object and agents who are brokers point to their best object that they do not broker. An object points to the agent who controls it. Cycles of this graph determine a submatching (agents receive the object they point to in a cycle). The partial matching, say $\nu$, determines the ownership rights in the next round $c_\nu$. For this algorithm to define a mechanism and for this mechanism to be efficient and group strategy-proof, it is necessary that $c$ satisfies further properties. Let $\Ccal$ be the set of inheritance structures that define an efficient and group strategy-proof mechanism. Two of these properties are relevant to our results.
\begin{enumerate}
  \item \textbf{Initial brokerage limits}. If there is more than one broker at the first step of the algorithm, the mechanism is a TC mechanism with three brokers. This implies that there are only three possible brokerage scenarios at the outset of the algorithm: (i) there are no brokers; (ii) there is one broker; or (iii) there are three brokers, $n=3$, and the mechanism is a TC mechanism with three brokers.
  \item \textbf{Ownership persistence}. Suppose that a submatching $\nu$ is calculated at some step of the algorithm for a given preference profile. If $i\in N_\nu^c$ owns  $x$ at the step in which a $\nu$ is determined, this ownerships persists, i.e., 
for all $\nu\subseteq \nu'\in {\cal S}$, $i\in N_{\nu'}^c$ implies $c_{\nu'}(x)=(i,o)$.
This implies that if an agent owns an object in some round of the algorithm, she will receive an object that is at least as good as this object.
\end{enumerate}

The next lemma states that if an agent is the owner of two objects at the beginning of the owner-and-broker algorithm, the mechanism is not balanced. Essentially, the issue is that if an agent owns two objects, at least an agent starts without being the owner or broker of an object. An agent who does not control any object at the outset of the algorithm will end up receiving her worst choice in some problems. This is not so for the agent who starts with control of more than two objects.

\begin{lemma}\rm\label{Lem:owner-two-houses}
Let $c\in\Ccal$. If there is an agent who owns more than one object at the first step of owner-and-broker algorithm with structure $c$, $TC_c$ is not balanced.
\end{lemma}

\begin{proof}
Let $i\in N$ be such that for two different objects, $x$ and $y$, $c_\emptyset(x)=(i,o)$ and $c_\emptyset(y)=(i,o)$. By ownership persistence, for any preference profile, agent $i$ will not be assigned her worst choice. However, when all agents have the same preferences, efficiency implies that each agent receives a distinct object. Let agent $j$ be the agent who receives her worst choice. Then, we have that agent $i$ never receives her worst object and there is another agent who receives her worst object for some profiles. Thus, $TC_c$ is not balanced.
\end{proof}

It remains to show that ownership and brokerage on average treat agents differently. This seems intuitive, but requires a relatively involved proof. The following lemma is the key to establish this fact. It states that  when all profiles are equally likely, the expected number of agents who receive their top choice, the second top choice, and so on, is the same for any two efficient and group strategy-proof mechanisms.

\begin{lemma}\rm \label{Lm-sum}Let $f$ and $g$ be two efficient and strategy-proof mechanisms. Then for each $k\in\{1,...,n\}$,
\[
\sum_{i\in N} \left|\left\{R\in\Rcal^N:f(R)_i=R_i^k \right\}\right|=\sum_{i\in N} \left|\left\{R\in\Rcal^N:g(R)_i=R_i^k \right\}\right|.\label{Eq*}
\]
\end{lemma}

\begin{proof}
Denote the set of all permutations $\pi:N\to N$ by $\Pi$.  Note that $|\Pi|=n!$. Given a mechanism $h:\Rcal^N\rightarrow\Mcal$ and $\pi\in\Pi$, the permuted form of $h$ with respect to $\pi$, denoted by $h^\pi$, is obtained by reassigning the role of agents in $h$ as determined by $\pi$.\footnote{$h^\pi$ corresponds to $\pi\bigodot h$ in \citet{Bade-MOR-2020}.} Formally, for each $R\in\Rcal^N$, let $\tau^\pi(R)\in\Rcal^N$ be the profile where for each $k\in N$, $\tau^\pi(R)_k$ coincides with $R_{\pi(k)}$, i.e., agent $k$ is endowed with the preference of agent $\pi(k)$ in profile $R$. (Note that $\tau^\pi:\Rcal^N\rightarrow \Rcal^N$ is a bijection.) Then, for each $R\in \Rcal^N$ and each $i\in N$, $h^\pi(R)_i:=h(\tau^\pi(R))_{\pi(i)}$.

Given a mechanism $h$ one can define its symmetrized random version by choosing a permutation $\pi$ uniformly at random from $\Pi$ and assigning each agent the corresponding allotment. Formally, the symmetrization of $h$ is the function $h^\Pi:\Rcal^N\rightarrow \Delta(\Mcal)$ defined as follows: for each $\mu\in\Mcal$,
\[h^\Pi(\mu):=|\{\pi\in\Pi:h^\pi(R)=\mu\}|/n!.\]

Fix $k\in\{1,...,n\}$. Theorem 1 in \citet{Bade-MOR-2020} states that any two efficient and group stratety-proof mechanisms induce the same symmetrized random mechanism. Thus, $f^\Pi=g^\Pi$ and  for each $R\in\Rcal^N$ and each $i\in N$,
\[\left|\left\{\pi\in\Pi:f^\pi(R)_i=R_i^k\right\}\right|/n!=
\left|\left\{\pi\in\Pi:g^\pi(R)_i=R_i^k\right\}\right|/n!.
\]
Equivalently,
\[\sum_{\pi\in\Pi}1_{f^\pi(R)_i=R_i^k}(i,R,\pi)=
\sum_{\pi\in\Pi}1_{g^\pi(R)_i=R_i^k}(i,R,\pi),\]
where $1_{f^\pi(R)_i=R_i^k}(i,R,\pi)=1$ if $f^\pi(R)_i=R_i^k$ and zero otherwise.  Thus,
\[\sum_{i\in N}\sum_{R\in\Rcal^N}\sum_{\pi\in\Pi}1_{f^\pi(R)_i=R_i^k}(i,R,\pi)=
\sum_{i\in N}\sum_{R\in\Rcal^N}\sum_{\pi\in\Pi}1_{g^\pi(R)_i=R_i^k}(i,R,\pi).
\]
Since the summations above are finite, we can exchange their order. Thus,
\begin{equation}\label{Eq:Eq2}
\sum_{\pi\in\Pi}\sum_{i\in N}\sum_{R\in\Rcal^N}1_{f^\pi(R)_i=R_i^k}(i,R,\pi)=
\sum_{\pi\in\Pi}\sum_{i\in N}\sum_{R\in\Rcal^N}1_{g^\pi(R)_i=R_i^k}(i,R,\pi).
\end{equation}
Now, fix a permutation $\pi\in\Pi$. By definition of $TC^\pi_c$, for each $R\in\Rcal^N$,
\[|\{i\in N:f^\pi(R)_i=R^k_i\}|=|\{i\in N:f(\tau^\pi(R))_i=\tau^\pi(R)^k_i\}|.\]
Equivalently,
\[\sum_{i\in N}1_{f^\pi(R)_i=R_i^k}(i,R,\pi)=
\sum_{i\in N}1_{f(\tau^\pi(R))_i=\tau^\pi(R)_i^k}(i,R),
\]
where $1_{f(\tau^\pi(R))_i=\tau^\pi(R)_i^k}(i,R)=1$ if $f(\tau^\pi(R))_i=\tau^\pi(R)_i^k$ and zero otherwise. Thus,
\[\sum_{R\in\Rcal^N}\sum_{i\in N}1_{f^\pi(R)_i=R_i^k}(i,R,\pi)=
\sum_{R\in\Rcal^N}\sum_{i\in N}1_{f(\tau^\pi(R))_i=\tau^\pi(R)_i^k}(i,R).
\]
Since $\tau^\pi:\Rcal^N\rightarrow \Rcal^N$ is a bijection
\[\sum_{R\in\Rcal^N}\sum_{i\in N}1_{f(\tau^\pi(R))_i=\tau^\pi(R)_i^k}(i,R)=\sum_{R\in\Rcal^N}\sum_{i\in N}1_{f(R)_i=R_i^k}(i,R),\]
where $1_{f(R)_i=R_i^k}(i,R)=1$ if $f(R)_i=R_i^k$ and zero otherwise. From these two last equations we get that,
\[\sum_{i\in N}\sum_{R\in\Rcal^N}1_{f^\pi(R)_i=R_i^k}(i,R,\pi)=
\sum_{i\in N}\sum_{R\in\Rcal^N}1_{f(R)_i=R_i^k}(i,R).
\]
Thus,
\begin{equation}\label{Eq:3}\sum_{\pi\in\Pi}\sum_{i\in N}\sum_{R\in\Rcal^N}1_{f^\pi(R)_i=R_i^k}(i,R,\pi)=
n!\sum_{i\in N}\sum_{R\in\Rcal^N}1_{f(R)_i=R_i^k}(i,R).
\end{equation}
Similarly,
\begin{equation}\label{Eq:4}\sum_{\pi\in\Pi}\sum_{i\in N}\sum_{R\in\Rcal^N}1_{g^\pi(R)_i=R_i^k}(i,R,\pi)=
n!\sum_{i\in N}\sum_{R\in\Rcal^N}1_{g(R)_i=R_i^k}(i,R).
\end{equation}
Then, (\ref{Eq:3}) and (\ref{Eq:4}) in (\ref{Eq:Eq2}) yield
\[\sum_{i\in N}\sum_{R\in\Rcal^N}1_{f(R)_i=R_i^k}(i,R)=\sum_{i\in N}\sum_{R\in\Rcal^N}1_{g(R)_i=R_i^k}(i,R).\]
Equivalently,
\[
\sum_{i\in N} \left|\left\{R\in\Rcal^N:f(R)_i=R_i^k \right\}\right|=\sum_{i\in N} \left|\left\{R\in\Rcal^N:g(R)_i=R_i^k \right\}\right|.
\]
\end{proof}

We can then establish that a broker is treated on average worse than an owner. We prove this by comparing a mechanism that has a broker at the outset with the mechanism in which this agent receives the object as endowment, all other equal. Since the change affects only this agent, one can easily show that the agent is treated better, on average, when starts as being an owner instead of a broker. We know that the expected number of agents who receive their top choice, their second top choice and so on is the same for these two mechanisms  (Lemma~\ref{Lm-sum}). Thus, it must be that, on average, when the agent is a broker, he or she is treated worse than some other agent.

\begin{lemma}\rm\label{Lem:TC-one-broker-nobalanced}
Let $c\in \Ccal$. If there is a single broker and $n-1$ owners at the first step of the owner-and-broker algorithm with structure $c$, $TC_c$ is not balanced.
\end{lemma}

\begin{proof}
Consider a TC mechanism with one broker $i$ who brokers $h_i$ and $n-1$ owners, i.e., each $j\in N\setminus \{i\}$ owns $h_j$. We denote this mechanism by  $TC$. Now, consider a TTC from individual endowments where each $i\in N$ initially owns object $h_i$. We denote this mechanism by $TTC$. We claim that
\begin{equation}\label{Eq:5}
\left\{R\in\Rcal^N:TC(R)_i=R_i^n \right\}\subsetneq\left\{R\in\Rcal^N:TTC(R)_i=R_i^n \right\}.
\end{equation}
Let $R\in\Rcal^N$ be such that $TC(R)_i=R_i^n$, i.e., $TC$ assigns the top choice to agent $i$ at profile $R$. There are two cases. If $TC(R)_i\neq h_i$, then at each step of the algorithm that computes $TC(R)$, agent $i$ points to her most preferred object. Thus, each cycle that forms in this instance of this algorithm is also a cycle in the algorithm that computes $TTC(R)$. Thus, $TC(R)=TTC(R)$. Now, if $TC(R)_i=h_i$, then $h_i=\tops(R_i,H)$ and $TTC(R)_i=h_i$. Thus,
\[
\left\{R\in\Rcal^N:TC(R)_i=R_i^n \right\}\subseteq\left\{R\in\Rcal^N:TTC(R)_i=R_i^n \right\}.
\]
Now, let $j\neq i$ and consider a profile $R$ in which all agents rank $h_i$ top and $h_j$ second. Clearly, $TTC(R)_i=h_i$ and $TC(R)_i=h_j$. Thus, (\ref{Eq:5}) holds. This implies that,
\[
\left|\left\{R\in\Rcal^N:TC(R)_i=R_i^n \right\}\right|<\left|\left\{R\in\Rcal^N:TTC(R)_i=R_i^n \right\}\right|.
\]
By Lemma~\ref{Lm-sum},
\[
\sum_{j\in N} \left|\left\{R\in\Rcal^N:TC(R)_j=R_j^n \right\}\right|=\sum_{j\in N} \left|\left\{R\in\Rcal^N:TCC(R)_j=R_j^n \right\}\right|.
\]
Thus, there is $j\in N\setminus\{i\}$ such that
\[
\left|\left\{R\in\Rcal^N:TC(R)_j=R_j^n \right\}\right|>\left|\left\{R\in\Rcal^N:TTC(R)_j=R_j^n \right\}\right|.
\]
Now, by Lemma~\ref{Lem:balanced}, $TTC$ is balanced. Thus,
\[\left|\left\{R\in\Rcal^N:TTC(R)_i=R_i^n \right\}\right|=\left|\left\{R\in\Rcal^N:TTC(R)_j=R_j^n \right\}\right|.\]
Thus,
\[
\left|\left\{R\in\Rcal^N:TC(R)_i=R_i^n \right\}\right|<\left|\left\{R\in\Rcal^N:TC(R)_j=R_j^n \right\}\right|.
\]
Thus, $TC$ is not balanced.
\end{proof}

\begin{proof}[Proof of Theorem~\ref{Th:Main}] Suppose that $f$ is balanced, efficient, and group strategy-proof. Then, there is $c\in \Ccal$ such that $f=TC_c$ \citep{Pycia-Unver-TE-2017,Bade-MOR-2020}. Since~$c$ satisfies initial brokerage limits, there are only three cases.

(i) There are no brokers at the first step of the owners-and-brokers algorithm associated with $c$: We claim that $TC_c$ is a TTC from individual endowments. Suppose by contradiction that this is not so. Thus, there is at least an agent who owns at least two objects at the first step of the owners-and-brokers algorithm associated with~$c$. By Lemma~\ref{Lem:owner-two-houses}, $TC_c$ is not balanced. This is a contradiction.

(ii) There is a single broker at the first step of the owners-and-brokers algorithm associated with $c$: We claim that there are $n-1$ owners at the first step of the owners-and-brokers algorithm associated with $c$. If this were not so, there would have to be at least an agent who owns two objects at the first step of the owners-and-brokers algorithm associated with $c$. By Lemma~\ref{Lem:owner-two-houses}, $TC_c$  would not be balanced. Since there are one broker and $n-1$ owners at the first step of the owners-and-brokers algorithm associated with $c$, by Lemma~\ref{Lem:TC-one-broker-nobalanced}, $TC_c$ is not balanced. This is a contradiction. Thus case (ii) never holds

(iii) There are three brokers and $TC_c$ is a TC mechanism with three brokers.

Thus, $TC_c$ is either a TTC from individual endowments or a TC mechanism with three brokers.

Lemma~\ref{Lem:balanced} implies that both TTC from individual endowments and TC mechanisms with three brokers are balanced. These mechanisms are efficient and group strategy-proof \citep{Pycia-Unver-TE-2017,Bade-MOR-2020}.
\end{proof}

The three axioms in Theorem~\ref{Th:Main} are independent. If one of them is violated, then there exists a mechanism that satisfies the other two. A serial dictatorship is efficient and group strategy-proof \citep{Svensson-1999-SCW}, but not balanced. One can easily see that a constant mechanism is balanced and group strategy-proof, but not efficient. The following mechanism is efficient and balanced, but not group strategy-proof.

\begin{example}\rm\label{Ex-tightness}Let $N=\{1,2,3\}$, $O=\{a,b,c\}$, initial endowment $\omega=(a,b,c)$, two preference profiles $\widehat R$ and $\widetilde R$, and the mechanism $\psi$ are given by
\begin{table}[ht]
\centering
\begin{tabular}{c c c}
$\widehat R_1$ & $\widehat R_2$ & $\widehat R_3$ \\
\hline
$\boxed{b}$ & $a$ & $\boxed{a}$\\
$c$ & $\boxed{c}$ & $c$\\
$a$ & $b$ & $b$
\end{tabular}
\hspace{1 cm}
\begin{tabular}{c c c}
$\widetilde R_1$ & $\widetilde R_2$ & $\widetilde R_3$ \\
\hline
$\boxed{c}$ & $\boxed{a}$ & $a$\\
$b$ & $b$ & $\boxed{b}$\\
$a$ & $c$ & $c$
\end{tabular}
\end{table}
\begin{align*}
\psi(R)=\left\{
\begin{array}{cl}
(b,c,a), & \textrm{if}\ R=\widehat R;\\
(c,a,b), & \textrm{if}\ R=\widetilde R;\\
TTC_\omega(R), &\textrm{otherwise}.
\end{array}
\right.
\end{align*}
The mechanism $\psi$ is efficient, for $TTC_\omega$ is efficient and the allocations for both $\widehat R$ and $\widetilde R$ are efficient for these profiles. It is balanced, because $TTC_\omega$ is balanced, and for each $i\in N$ and each $k\in\{1,2,3\}$, $|\{R\in\{\widehat R,\widetilde R\}:\psi(R)_i=R_i^k\}|=|\{R\in\{\widehat R,\widetilde R\}:TTC(R)_i=R_i^k\}|$. However, $\psi$ is not group strategy-proof. If $\overline R_2=aP_2bP_2c$, then $\psi_2(\widehat R_{-2}, \overline R_2)\widehat P_2\psi_2(\widehat R)$. $\triangle$
\end{example}

\section{Concluding remarks}\label{Sec.COnclusion}

We have articulated a fairness notion for the house allocation problem. Our axiom, balancedness, requires a mechanism to assign the top choice, the second top choice, and so on, for the same number of profiles to all agents. This axiom guarantees symmetric  treatment at the stage in which the mechanism is announced and preferences are not reported yet. It implicitly assumes that all preferences profiles are equally likely.

The main message of our paper is that with an interesting exception in the three-agent case, the Top Trading Cycles from individual endowments are the only mechanisms that are balanced and satisfy efficiency and group strategy-proofness. Thus, Top Trading Cycles is the only mechanism that avoids the perception of bias by a mechanism designer at the stage in which the mechanism is announced and  also has desirable efficiency and incentives properties.

It is interesting but outside the scope of our paper to extend our notion of balancedness to school choice environments. In this context it is usually assumed that students are indifferent among all seats in a given school. If schools have different capacities, Top Trading Cycles is not balanced any longer. It is an open question to determine whether an efficient and group strategy-proof mechanism can be singled out as optimal in this class, with respect to the balancedness criterion, or at least unambiguously improves with respect to the Top Trading Cycles. A similar question applies to environments in which preference profiles are not all equally probable (and its probability is not symmetric as in \citealp{Freeman-et-al-2021}).

\bibliography{ref-BHA}

\end{document}